\documentclass[11pt,letterpaper,nosumlimits]{amsart}
\usepackage{url}
\usepackage{amssymb}
\usepackage{graphicx, color}
\usepackage{mathtools}
\usepackage{enumerate}
\usepackage[margin=1in]{geometry}

\newtheorem{theorem}{Theorem}[section]
\newtheorem{proposition}[theorem]{Proposition}
\newtheorem{lemma}[theorem]{Lemma}

\theoremstyle{definition}

\theoremstyle{remark}

\DeclareMathOperator{\tr}{tr}
\DeclareMathOperator{\rank}{rank}

\newcommand{\tp }{{\scriptscriptstyle\mathsf{T}}}

\renewcommand{\le}{\leqslant}
\renewcommand{\ge}{\geqslant}

\usepackage[normalem]{ulem}
\usepackage{contour}

\contourlength{1pt}

\newcommand{\myuline}[1]{%
  \uline{\phantom{#1}}%
  \llap{\contour{white}{#1}}%
}

\begin{document}
\title{Grothendieck constant is norm of Strassen matrix multiplication tensor}
\author{Jinjie~Zhang}
\address{Department of Statistics, University of Chicago, Chicago, IL, 60637-1514.}
\email{jinjie@galton.uchicago.edu}
\author{Shmuel~Friedland}
\address{Department of Mathematics, Statistics and Computer Science,  University of Illinois, Chicago, IL,  60607-7045.}
\email{friedlan@uic.edu}
\author{Lek-Heng~Lim}
\thanks{This article provides the details for Slide 16 in \url{http://www.ub.edu/focm2017/slides/Lim.pdf}, presented during the Smale Prize Lecture of the 2017 FoCM conference.}
\address{Computational and Applied Mathematics Initiative, Department of Statistics,
University of Chicago, Chicago, IL 60637-1514.}
\email[corresponding author]{lekheng@galton.uchicago.edu}

\begin{abstract}
We show that two important quantities from two disparate areas of complexity theory --- Strassen's exponent of matrix multiplication $\omega$ and Grothendieck's constant $K_G$
--- are intimately related. They are different measures of size for the same underlying object --- the matrix multiplication tensor, i.e., the $3$-tensor or bilinear operator $\mu_{l,m,n} : \mathbb{F}^{l \times m} \times \mathbb{F}^{m \times n} \to \mathbb{F}^{l \times n}$, $(A,B) \mapsto AB$ defined by matrix-matrix product over $\mathbb{F} = \mathbb{R}$ or $\mathbb{C}$. It is well-known that Strassen's exponent of matrix multiplication is the greatest lower bound on (the log of) a \emph{tensor rank} of $\mu_{l,m,n}$. We will show that Grothendieck's constant is the least upper bound on a \emph{tensor norm} of $\mu_{l,m,n}$, taken over all $l, m, n \in \mathbb{N}$. Aside from relating the two celebrated quantities, this insight allows us to rewrite Grothendieck's inequality as a norm inequality
\[
\lVert\mu_{l,m,n}\rVert_{1,2,\infty} =\max_{X,Y,M\neq0}\frac{|\tr(XMY)|}{\lVert X\rVert_{1,2}\lVert Y\rVert_{2,\infty}\lVert M\rVert_{\infty,1}}\le  K_G.
\]
We prove that Grothendieck's inequality is unique: If we generalize the $(1,2,\infty)$-norm to arbitrary $p,q, r \in [1, \infty]$,
\[
\lVert\mu_{l,m,n}\rVert_{p,q,r}=\max_{X,Y,M\neq0}\frac{|\tr(XMY)|}{\|X\|_{p,q}\|Y\|_{q,r}\|M\|_{r,p}},
\]
then $(p,q,r )=(1,2,\infty)$ is, up to cyclic permutations, the only choice for which  $\lVert\mu_{l,m,n}\rVert_{p,q,r}$ is uniformly bounded by a constant independent of $l,m,n$.
\end{abstract}



\maketitle

\section{Introduction}\label{sec:intro}
Grothendieck's inequality was originally established to relate fundamental norms on tensor product spaces \cite{Grothendieck}. Throughout this article, we will let $\mathbb{F} = \mathbb{R}$ or $\mathbb{C}$. The \emph{Grothendieck constant} $K_G^\mathbb{F}$ is the sharp constant such that for every $l,m,n\in\mathbb{N}$ and every matrix $M=(M_{ij})\in\mathbb{F}^{m\times n}$,
\begin{equation}\label{GI}
\max_{\lVert x_i\rVert = \lVert y_j \rVert = 1}\Bigl| \sum_{i=1}^m\sum_{j=1}^n M_{ij} \langle x_{i}, y_{j}\rangle\Bigr|\le  K_G^\mathbb{F}\max_{\lvert \varepsilon_i \rvert = \lvert \delta_j \rvert = 1}\Bigl|\sum_{i=1}^m\sum_{j=1}^n M_{ij} \varepsilon_i \delta_j\Bigr|
\end{equation}
where the maximum on the left is take over all $x_i,y_j\in\mathbb{F}^{l}$ of unit $2$-norm, and the maximum on the right is taken over all  $\varepsilon_i, \delta_j \in \mathbb{F}$ of unit absolute value (so over $\mathbb{R}$, $\varepsilon_i = \pm 1$ and  $\delta_j=\pm1$; over $\mathbb{C}$, $\varepsilon_i = e^{i\theta_i}$ and  $\delta_j= e^{i \phi_j}$). The value on the left side of \eqref{GI} is the same for all $l\ge  m+n$ and as such some authors restrict themselves to  $l = m+n$.

The existence of a such a constant independent of $l$, $m$ and $n$  was discovered by Alexandre Grothendieck in $1953$. Alternative proofs via factorization of linear operators, geometry of Banach spaces, absolutely $p$-summing operators, etc, may be found in \cite{Pisier, Jameson, Lindenstrauss,Pisier2} and references therein. In particular, the formulation in \eqref{GI} was due to Lindenstrauss and Pe{\l}czy\'nski \cite{Lindenstrauss}.

The inequality has found  applications in numerous areas, including Banach space theory, $C^{*}$ algebra, harmonic analysis, operator theory, quantum mechanics, and most recently, computer science. In theoretical computer science, Grothendieck's inequality has notably appeared in studies of unique games conjecture \cite{Khot, Khot2, Kindler, Rag, Rag2} and SDP relaxations of NP-hard combinatorial problems \cite{Alon1, Alon2, Alon3, Arora, Charikar}. In quantum information theory, Grothendieck's inequality arises unexpectedly in Bell inequalities  \cite{Fishburn, Tsi1, Heydari} and in XOR games \cite{Bri1, Bri2, Bri3}, among several other areas; Grothendieck constants of specific orders, e.g., $K_G^\mathbb{C}(3)$ and $K_G^\mathbb{C}(4)$, also have important roles to play in quantum information theory \cite{Acin, Hirsch, Bene}. The inequality has even been applied to some rather surprising areas, e.g.,  to communication complexity \cite{Linial, Reg1, Reg2} and to privacy-preserving data analysis \cite{Talwar}.

Although the Grothendieck constant appears in numerous mathematical statements and has many equivalent interpretations in physics and computer science, its exact value  remains unknown and estimating increasingly sharper bounds for $K_G^\mathbb{F}$ has been a major undertaking. The current best known bounds are $K_G^\mathbb{R}\in[1.676, 1.782]$, established in \cite{Davie} (lower) and \cite{Krivine} (upper); and $K_G^\mathbb{C}\in(1.338, 1.404]$, established in \cite{Davie2} (lower) and \cite{Haagerup} (upper). A major recent breakthrough  \cite{Braverman} established that Krivine's upper bound $\pi/\bigl(2\log(1+\sqrt{2})\bigr)\approx 1.782$ for $K_G^\mathbb{R}$ is not sharp. There have also been efforts in approximating Grothedieck's constants of specific orders, e.g., see \cite{Hirsch, Bene} for recent results on $K_G^\mathbb{C}(3)$ and $K_G^\mathbb{C}(4)$.

A world apart from the aforementioned areas touched by Grothendieck's inequality is the problem of complexity of matrix inversion, or equivalently, matrix multiplication, pioneered by  Volker Strassen \cite{Strass, Strass0, Strass1, Strass2}. A systematic study of this and other related problems has blossomed into what is now often called algebraic computational complexity \cite{BCS}. For the uninitiated, Strassen famously discovered in \cite{Strass} that the product of a pair of $2 \times 2$ matrices may be obtained with just seven multiplications:
\[
\begin{bmatrix}
a_1 & a_2 \\
a_3 & a_4
\end{bmatrix}
\begin{bmatrix}
b_1 & b_2 \\
b_3 & b_4
\end{bmatrix}
=
\begin{bmatrix}
a_1 b_1 + a_2 b_2 & \beta +\gamma +  (a_1+a_2 -a_3 -a_4)b_4\\
\alpha +\gamma+a_4(b_2+b_3 -b_1 -b_4) &  \alpha + \beta + \gamma
\end{bmatrix},
\]
where $\alpha = (a_3-a_1)(b_3 -b_4)$, $\beta = (a_3+a_4)(b_3- b_1)$, $\gamma =a_1 b_1 + (a_3+a_4 -a_1)(b_1 + b_4 - b_3)$. Applied recursively, this gives an algorithm for forming the product of a pair of $n \times n$ matrices with just  $O(n^{\log_2 7})$ multiplications, as opposed to $O(n^3)$ using the usual formula for matrix-matrix product. In addition, Strassen also showed that: (i) the number of additions may be bounded by a constant times the number of multiplications; (ii) matrix inversion may be achieved with the same complexity as matrix multiplication. In short, if there is an algorithm that forms matrix product in $O(n^\omega)$ multiplication, there it yields a $O(n^\omega)$  algorithm  that would solve $n$ linear equations in $n$ unknowns, which is by far the most ubiquitous problem in all of scientific and engineering computing. The smallest possible $\omega$ became known as the \emph{exponent of matrix multiplication}.

Strassen's astounding discovery captured the interests of numerical analysts and theoretical computer scientists alike and the complexity was gradually lowered over the years. Some milestones include the Coppersmith--Winograd \cite{CW} bound $O(n^{2.375477})$ that resisted progress for more than two decades until Vassilevska-Williams's improvement \cite{Williams} to $O(n^{2.3728642})$; the current record, due to Le Gall \cite{LeGall}, is $O(n^{2.3728639})$. Strassen showed \cite{Strass2} that the best possible $\omega$ is in fact given by
\[
\omega =\inf_{n \in \mathbb{N}} \log_n \bigl(\rank(\mu_{n,n,n})\bigr),
\]
where $\mu_{n,n,n}$  is the \emph{Strassen matrix multiplication tensor} --- the $3$-tensor in $(\mathbb{F}^{n \times n} )^* \otimes (\mathbb{F}^{n \times n} )^* \otimes \mathbb{F}^{n \times n} $ associated with matrix-matrix product, i.e., the bilinear operator
\[
\mathbb{F}^{n \times n} \times  \mathbb{F}^{n \times n} \to  \mathbb{F}^{n \times n}, \qquad (A, B) \mapsto AB.
\]
Those unfamiliar with multilinear algebra \cite{Lang} may regard the $3$-tensor $\mu_{n,n,n}$ and the bilinear operator as the same object. If we choose a basis on $\mathbb{F}^{n \times n}$ (or three different bases, one on each copy of  $\mathbb{F}^{n \times n}$), then $\mu_{n, n, n}$ may be represented as a $3$-dimensional hypermatrix in $\mathbb{F}^{n^2 \times n^2 \times n^2}$. Over any $\mathbb{F}$-vector spaces $\mathbb{U}$, $\mathbb{V}$, $\mathbb{W}$, one may define \emph{tensor rank} \cite{Hi1} for $3$-tensors $\tau \in \mathbb{U} \otimes \mathbb{V} \otimes \mathbb{W}$ by
\[
\rank(\tau) =\min\left\{r :
\tau=\sum_{i=1}^{r} \lambda_i u_{i}\otimes v_{i}\otimes w_{i}\right\}.
\]
In fact, Strassen showed that the tensor rank of a  $3$-tensor $\mu_\beta \in \mathbb{U}^* \otimes \mathbb{V}^* \otimes \mathbb{W}$  associated with a bilinear operator  $\beta : \mathbb{U} \times \mathbb{V} \to \mathbb{W}$ gives the least number of multiplications required to compute $\beta$. The value of $\omega$ is in general dependent on the choice of $\mathbb{F}$, as tensor rank is well-known to be field dependent  \cite{HLA}.

What exactly is $\omega$? The above discussion shows that it is the sharp lower bound for the (log of the) tensor rank of the Strassen matrix multiplication tensor:
\begin{equation}\label{eq:strass1}
\log_n\bigl( \rank(\mu_{n,n,n}) \bigr) \ge  \omega \qquad \text{for all}\; n \in \mathbb{N}.
\end{equation}

What exactly is $K_G^\mathbb{F}$? We will show that it is the sharp upper bound for the tensor $(1,2,\infty)$-norm of the Strassen matrix multiplication tensor:
\begin{equation}\label{eq:gro1}
\lVert\mu_{l,m,n}\rVert_{1,2,\infty} \le  K_G^\mathbb{F}\qquad \text{for all} \; l, m, n \in \mathbb{N}.
\end{equation}
If we desire a greater parallel to \eqref{eq:strass1}, we may drop $l$ and $m$ in \eqref{eq:gro1} --- there is no loss of generality in assuming that $l=2n$ and $m = n$, i.e., $K_G^\mathbb{F}$ is also the sharp upper bound so that
\[
\lVert\mu_{2n,n,n}\rVert_{1,2,\infty} \le  K_G^\mathbb{F} \qquad \text{for all} \; n \in \mathbb{N}.
\]
In addition, the \emph{Grothendieck constant of order $l \in \mathbb{N}$}, a popular notion in quantum information theory (e.g., \cite{Acin, Bene, Hirsch}), is given by a simple variation, namely, the sharp upper bound $ K_G^\mathbb{F} (l) $ in
\[
\lVert\mu_{l,m,n}\rVert_{1,2,\infty} \le K_G^\mathbb{F} (l) \qquad \text{for all} \;  m, n \in \mathbb{N}.
\]

We will define the $(1,2,\infty)$-norm for an arbitrary $3$-tensor formally in Section~\ref{sec:ineq} but at this point it suffices to know its value for $\mu_{l,m,n}$, namely,
\[
\lVert\mu_{l,m,n}\rVert_{1,2,\infty} =\max_{X,Y,M\neq0}\frac{|\tr(XMY)|}{\lVert X\rVert_{1,2}\lVert Y\rVert_{2,\infty}\lVert M\rVert_{\infty,1}}
\]
where $X\in\mathbb{F}^{l\times m}$, $M\in\mathbb{F}^{m\times n}$, $Y\in\mathbb{F}^{n\times l}$, and $\lVert M \rVert_{p,q} \coloneqq  \max_{x \ne 0} \lVert Mx \rVert_q /\lVert x \rVert_p$ denotes the matrix $(p,q)$-norm.

The inequality \eqref{eq:gro1} is in fact just Grothendieck's inequality.  The characterizations of $\omega$ and $K_G$ in \eqref{eq:strass1} and \eqref{eq:gro1} hold over both $\mathbb{R}$ and $\mathbb{C}$ although their values are field dependent. Incidentally the fact that  Grothendieck's constant is essentially a tensor norm immediately explains why  it is field dependent --- because, as is the case for tensor rank, tensor norms are also field dependent \cite{nuclear}.

An advantage of the formulation in \eqref{eq:gro1} is that we obtain  a natural family of $(p,q,r)$-norms on $\mu_{l,m,n}$ given by
\[
\lVert\mu_{l,m,n}\rVert_{p,q,r}\coloneqq \max_{X,Y,M\neq0}\frac{|\tr(XMY)|}{\lVert X\rVert_{p,q}\lVert Y\rVert_{q,r}\lVert M\rVert_{r,p}}
\]
for any  triple $1 \le p,q,r \le \infty$. This family of norms will serve as  a platform for us to better comprehend Grothendieck's constant and Grothendieck's inequality. The study of the general $(p,q,r)$-case shows why the $(1,2,\infty)$-case is extraordinary. We will deduce a generalization of Grothendieck's inequality and show that the case $(p,q,r) = (1,2,\infty)$, i.e., Grothendieck's inequality, is the only one up to trivial cyclic permutations\footnote{Unavoidable as $(p,q,r)$-norms are clearly invariant under cyclic permutations of $p,q,r$. See Lemma~\ref{lem:simple}\eqref{it:cyclic}.} where there is a universal upper bound, i.e., Grothendieck's constant, that holds for all $l,m,n \in \mathbb{N}$.
\begin{theorem}[Grothendieck--H\"older inequality]\label{thm:GroHol}
Let $1 \le p,q,r\le \infty$ and $l,m,n\in\mathbb{N}$. Then
\[
\frac{1}{l^{|1/q-1/2|}\cdot m^{|1/p-1/2|}\cdot n^{|1/r-1/2|}} \le
\lVert\mu_{l,m,n}\rVert_{p,q,r}\le  K_G^\mathbb{F} \cdot l^{|1/q-1/2|} \cdot m^{1-1/p} \cdot n^{1/r}.
\]
In particular, when $p=1$, $q=2$, and $r=\infty$, the upper bound gives Grothendieck's inequality \eqref{GI}.
\end{theorem}
\begin{theorem}[Uniqueness of Grothendieck's inequality]\label{thm:unique1}
Let $1 \le p,q,r\le \infty$ and $l,m,n\in\mathbb{N}$. Then $\lVert\mu_{l,m,n}\rVert_{p,q,r}$ is uniformly bounded for all $l, m,n \in \mathbb{N}$ if and only if
\[
(p,q,r) \in \{ (1,2,\infty), \; (\infty, 1, 2), \; (2, \infty, 1)\}.
\]
\end{theorem}
Theorem~\ref{thm:GroHol} follows from Theorems~\ref{thm:upper} and \ref{thm:lower}. Theorem~\ref{thm:unique1} is just Theorem~\ref{thm:unique}.

\section{Strassen matrix multiplication tensor}\label{sec:struct}

An important observation for us, obvious to anyone familiar with tensors \cite{Landsberg, Lang, HLA} but perhaps less so to those accustomed to regarding (erroneously) a tensor as a ``multiway array,'' is that the bilinear operator
\begin{equation}\label{eq:mm}
\beta \in \mathbb{F}^{l \times m} \times  \mathbb{F}^{m \times n} \to  \mathbb{F}^{l \times n}, \qquad (X, Y) \mapsto XY,
\end{equation}
and the trilinear functional
\begin{equation}\label{eq:tr}
\tau : \mathbb{F}^{l \times m} \times  \mathbb{F}^{m \times n} \times  \mathbb{F}^{n \times l} \to \mathbb{F}, \qquad (X, Y, Z) \mapsto \tr(XYZ),
\end{equation}
are given by\footnote{To be more precise, by the universal property of tensor products \cite[Chapter~XVI, \S1]{Lang}, $\beta$ induces a linear map $\beta_* : \mathbb{F}^{l \times m} \otimes  \mathbb{F}^{m \times n} \to  \mathbb{F}^{l \times n}$ and $\tau$ induces a linear map $\tau_* :  \mathbb{F}^{l \times m} \otimes  \mathbb{F}^{m \times n} \otimes  \mathbb{F}^{n \times l} \to \mathbb{F}$, i.e., $\beta_* \in (\mathbb{F}^{l \times m})^* \otimes  (\mathbb{F}^{m \times n})^* \otimes  \mathbb{F}^{l \times n}$ and $\tau_* \in (\mathbb{F}^{l \times m})^* \otimes  (\mathbb{F}^{m \times n})^* \otimes (\mathbb{F}^{n \times l})^*$. We identify $\beta, \tau$ with the linear maps $\beta_*, \tau_*$ they induce.} the same $3$-tensor in
\[
 (\mathbb{F}^{l \times m})^* \otimes  (\mathbb{F}^{m \times n})^* \otimes  \mathbb{F}^{l \times n} \cong (\mathbb{F}^{l \times m})^* \otimes  (\mathbb{F}^{m \times n})^* \otimes  (\mathbb{F}^{n \times l})^*.
\]
In other words, as $3$-tensors, there is no difference between the product of two matrices and the trace of product of three matrices.

To see this, let $E_{ij} \in \mathbb{F}^{m \times n}$ denote the matrix with $1$ in its $(i,j)$th entry and zeros everywhere else, so that $\{ E_{ij} : i=1,\dots, m; \; j =1,\dots, n\}$ is the standard basis for $\mathbb{F}^{m \times n}$. Its dual basis for the dual space of linear functionals
\[
(\mathbb{F}^{m \times n})^* \coloneqq  \{ \varphi : \mathbb{F}^{m \times n} \to \mathbb{F} : \varphi(\alpha X + \beta Y) = \alpha \varphi(X) + \beta \varphi(Y) \}
\]
is then given by  $\{ \varepsilon_{ij} : i=1,\dots, m; \; j =1,\dots, n\}$ where $\varepsilon_{ij}:\mathbb{F}^{m\times n}\rightarrow\mathbb{F}$, $X \mapsto x_{ij}$, is the linear
functional that takes an $m\times n$ matrix to its $(i,j)$th entry.
Now choose the standard inner product on $\mathbb{F}^{m \times n}$, i.e., $\langle X, Y\rangle = \tr(X^\tp Y)$. Then $\varepsilon_{ij}(X) = \langle E_{ij}, X \rangle$ for all $X \in \mathbb{F}^{m \times n}$,
which allows us to identify $(\mathbb{F}^{ m \times n})^*$ with $\mathbb{F}^{n \times m}$ and linear functional $\varepsilon_{ij} \in( \mathbb{F}^{ m \times n})^* $ with the matrix $E_{ji} \in \mathbb{F}^{n \times m}$.

It remains to observe that the usual formula for matrix-matrix product gives
\begin{align*}
\beta(X,Y) &= \sum_{i,k=1}^{l,n} \left(\sum_{j=1}^m x_{ij} y_{jk} \right) E_{ik}
 = \sum_{i,k=1}^{l,n} \left(\sum_{j=1}^m \varepsilon_{ij}(X) \varepsilon_{jk}(Y) \right) E_{ik} \\
 & = \sum_{i,k=1}^{l,n} \left(\sum_{j=1}^m (\varepsilon_{ij} \otimes \varepsilon_{jk}) (X, Y) \right) E_{ik}
 =  \left(\sum_{i,j,k=1}^{l,m,n} \varepsilon_{ij} \otimes \varepsilon_{jk}  \otimes  E_{ik} \right) (X, Y),
\end{align*}
and thus
\begin{equation}\label{eq:beta}
\beta = \sum_{i,j,k=1}^{l,m,n} \varepsilon_{ij} \otimes \varepsilon_{jk}  \otimes  E_{ik} \in
 (\mathbb{F}^{l \times m})^* \otimes  (\mathbb{F}^{m \times n})^* \otimes  \mathbb{F}^{l \times n} .
\end{equation}
A similar simple calculation,
\begin{align*}
\tau(X, Y, Z)&  =\sum_{i,j,k=1}^{l,m,n} x_{ij} y_{jk} z_{ki} =\sum_{i,j,k=1}^{l,m,n} \varepsilon_{ij}(X)\varepsilon_{jk}(Y)\varepsilon_{ki}(Z)\\
&  = \sum_{i,j,k=1}^{l,m,n} (\varepsilon_{ij}\otimes \varepsilon_{jk}\otimes\varepsilon_{ki}) (X,Y,Z) = \left(\sum_{i,j,k=1}^{l,m,n} \varepsilon_{ij}\otimes \varepsilon_{jk}\otimes\varepsilon_{ki} \right) (X,Y, Z) ,
\end{align*}
gives
\begin{equation}\label{eq:tau}
\tau = \sum_{i,j,k=1}^{l,m,n} \varepsilon_{ij}\otimes \varepsilon_{jk}\otimes\varepsilon_{ki}  \in
 (\mathbb{F}^{l \times m})^* \otimes  (\mathbb{F}^{m \times n})^* \otimes  (\mathbb{F}^{n \times l})^*.
\end{equation}
By our identification, $(\mathbb{F}^{ m \times n})^*= \mathbb{F}^{n \times m}$ and $\varepsilon_{ki} = E_{ik}$. So we see from \eqref{eq:beta} and \eqref{eq:tau} that indeed $\beta = \tau$ as $3$-tensors. We denote this tensor by $\mu_{l,m,n}$. This has been variously called the \emph{Strassen matrix multiplication tensor} or the \emph{structure tensor for matrix-matrix product} \cite{BCS, Landsberg, HLA, YL}.

\section{Grothendieck's constant and Strassen's tensor}\label{sec:gro}

Let $l,m,n$ be positive integers and let $M = (M_{ij})\in \mathbb{F}^{m\times n}$. Let  $x_{1},\dots,x_{m}, y_{1},\dots,y_{n} \in\mathbb{F}^l$ be vectors of unit $2$-norm. We will regard $x_{1},\dots,x_{m}$ as columns of a matrix $X\in\mathbb{F}^{l\times m}$ and $y_{1}^\tp,\dots,y_{n}^\tp$ as rows of a matrix $Y\in\mathbb{F}^{n\times l}$.

Recall that for any $p\ge  1$ with  H\"older conjugate $p^*$, i.e., $1/p+ 1/p^*=1$, we have
\begin{equation}\label{eq:pqnorm}
\|X\|_{1,p} \coloneqq \max_{z\neq0}\frac{\|Xz\|_p}{\|z\|_1}=\max_{i=1,\dots, m}\|x_i\|_p,\qquad
\|Y\|_{p,\infty}\coloneqq \max_{z\neq0}\frac{\|Yz\|_\infty}{\|z\|_p}=\max_{i=1,\dots, n}\|y_i\|_{p^*},
\end{equation}
and
\[
\|M\|_{\infty,1}\coloneqq \max_{z\neq0}\frac{\|Mz\|_{1}}{\|z\|_{\infty}}
=\max_{|\delta_j|=1}\sum_{i=1}^{m}\left|\sum_{j=1}^{n}M_{ij}\delta_j \right|
=\max_{|\varepsilon_i|=1, \; |\delta_j|=1}\left|\sum_{i=1}^m\sum_{j=1}^n M_{ij} \varepsilon_i \delta_j \right|,
\]
which may be further simplified for $\mathbb{F}=\mathbb{R}$ as
\begin{equation}\label{eq:inftyone}
\|M\|_{\infty,1} =\max_{\varepsilon_i =\pm1, \; \delta_j=\pm1}\left|\sum_{i=1}^m\sum_{j=1}^n M_{ij} \varepsilon_i \delta_j \right| =  \max_{\varepsilon, \delta \in \{\pm 1\}^n} \lvert \varepsilon^\tp M \delta \rvert.
\end{equation}
We refer the reader to \cite{nuclear} for a proof that
\begin{equation}\label{eq:norm}
\lVert \tau \rVert_{1,2,\infty} \coloneqq  \max_{X,Y,M\neq0}\frac{|\tau(X, M, Y)|}{\|X\|_{1,2}\|Y\|_{2,\infty}\|M\|_{\infty,1}}
\end{equation}
defines a norm for any tensor $\tau \in  (\mathbb{F}^{l \times m})^* \otimes  (\mathbb{F}^{m \times n})^* \otimes  (\mathbb{F}^{n \times l})^*$, regarded as a trilinear functional.

Since
\[
\sum_{i=1}^{m}\sum_{j=1}^n M_{ij}\langle x_i,y_j\rangle=
\begin{cases}
\tr(XMY) &\text{if} \; \mathbb{F} = \mathbb{R}, \\
\tr(XM\overline{Y})&\text{if} \; \mathbb{F} = \mathbb{C},
\end{cases}
\]
we see that Grothendieck's inequality \eqref{GI} may be stated as
\begin{equation}\label{Gro_norm_combine}
\max_{X,Y,M\neq0}\frac{|\tr(XMY)|}{\|X\|_{1,2}\|Y\|_{2,\infty}\|M\|_{\infty,1}}\le  K_G^\mathbb{F},
\end{equation}
when $\mathbb{F} = \mathbb{R}$ and as
\[
\max_{X,Y,M\neq0}\frac{|\tr(XM\overline{Y})|}{\|X\|_{1,2}\|Y\|_{2,\infty}\|M\|_{\infty,1}}\le  K_G^\mathbb{F},
\]
when $\mathbb{F} = \mathbb{C}$. However, in the latter case, we may write
\[
\max_{X,Y,M\neq0} \frac{|\tr(XM\overline{Y})|}{\|X\|_{1,2}\|Y\|_{2,\infty}\|M\|_{\infty,1}} = \max_{X,\overline{Y},M\neq0} \frac{|\tr(XM\overline{Y})|}{\|X\|_{1,2}\|\overline{Y}\|_{2,\infty}\|M\|_{\infty,1}} = \max_{X,Y,M\neq0} \frac{|\tr(XMY)|}{\|X\|_{1,2}\|Y\|_{2,\infty}\|M\|_{\infty,1}}
\]
as matrix $(p,q)$-norms are invariant under complex conjugation. Hence \eqref{Gro_norm_combine} in fact gives Grothendieck's inequality  for both $\mathbb{F} = \mathbb{R}$ and $\mathbb{C}$.
By our discussion in Section~\ref{sec:struct} and our norm in \eqref{eq:norm}, \eqref{Gro_norm_combine} is just
\[
\lVert \mu_{l,m,n} \rVert_{1,2,\infty} \le K_G^\mathbb{F}
\]
where
\begin{equation}\label{tensor_mu}
\mu_{l,m,n}\coloneqq \sum_{i=1}^{l}\sum_{j=1}^{m}\sum_{k=1}^{n}
\varepsilon_{ij}\otimes\varepsilon_{jk}\otimes\varepsilon_{ki}\in\bigl(  \mathbb{F}^{l\times m}\bigr)
^{\ast}\otimes\bigl(  \mathbb{F}^{m\times n}\bigr)  ^{\ast}\otimes\bigl(
\mathbb{F}^{n\times l }\bigr)  ^{\ast}
\end{equation}
is the Strassen matrix multiplication tensor for the product of $l \times m$ and $m \times n$ matrices.

This allows us to define Grothendieck's constant in terms of tensor norms: For $\mathbb{F} = \mathbb{R}$ or $\mathbb{C}$,
\[
K_G^\mathbb{F} = \sup_{l,m,n\in\mathbb{N}}\lVert\mu_{l,m,n}\rVert_{1, 2, \infty}.
\]
Since $\lVert\mu_{l,m,n}\rVert_{1, 2, \infty} = \lVert\mu_{m+n,m,n}\rVert_{1, 2, \infty}$ for all $l \ge m+n$,
\[
K_G^\mathbb{F} = \sup_{m,n\in\mathbb{N}}\lVert\mu_{m+n,m,n}\rVert_{1, 2, \infty} = \sup_{n\in\mathbb{N}}\lVert\mu_{2n,n,n}\rVert_{1, 2, \infty}.
\]
In addition, the  Grothendieck constant of order $l \in \mathbb{N}$ \cite{Acin, Bene, Hirsch} may be defined as
\[
K_G^\mathbb{F} (l) =\sup_{m,n\in\mathbb{N}}\lVert\mu_{l,m,n}\rVert_{1,2,\infty}.
\]

\section{Grothendieck--H\"older inequality}\label{sec:ineq}

The  norm  in \eqref{eq:norm} admits a natural generalization to arbitrary $p,q,r\in[1,\infty]$ as
\[
\lVert \tau \rVert_{p,q,r}\coloneqq \max_{X,Y,M\neq0}\frac{|\tau(X,M,Y)|}{\lVert X\rVert_{p,q}\lVert Y\rVert_{q,r}\lVert
M\rVert_{r,p}}
\]
defined for any $\tau \in  (\mathbb{F}^{l \times m})^* \otimes  (\mathbb{F}^{m \times n})^* \otimes  (\mathbb{F}^{n \times l})^*$, regarded as a trilinear functional
\[
\tau :  \mathbb{F}^{l \times m} \times  \mathbb{F}^{m \times n} \times  \mathbb{F}^{n \times l} \to \mathbb{F}.
\]
In this article, we will only be interested in $\tau = \mu_{l,m,n}$, the Strassen tensor. We first state some simple observations that will be useful later.
\begin{lemma}\label{lem:simple}
Let $p,q,r\in[1,\infty]$. Then  the $(p,q,r)$-norm of $ \mu_{l,m,n}$
\begin{enumerate}[\upshape (i)]
\item\label{it:cyclic}  is invariant under cyclic permutation of $p,q, r$,
\[
\lVert \mu_{l,m,n}\rVert_{p,q,r} =\lVert \mu_{l,m,n}\rVert_{r,p,q} = \lVert \mu_{l,m,n}\rVert_{q,r,p};
\]
\item\label{it:conjugate} transforms under H\"older conjugation as
\[
\lVert \mu_{l,m,n}\rVert_{p,q,r} = \lVert \mu_{l,m,n}\rVert_{r^*,q^*,p^*}.
\]
Recall that $p^*$ is the H\"older conjugate of $p$, i.e., $1/p + 1/p^* = 1$.
\end{enumerate}
\end{lemma}
\begin{proof}
Since the numerator $\tr(XMY) = \tr(MYX) = \tr(YXM)$ and the denominator is the product $\|X\|_{p,q}\|M\|_{r,p}\|Y\|_{q,r}$, cyclic permutations of  $(p,q)$,  $(r,p)$, $(q,r)$ leave the quotient
\[
\frac{|\tr(XMY)|}{\|X\|_{p,q}\|M\|_{r,p}\|Y\|_{q,r}}
\]
invariant. Now just observe that the cyclic permutations
\[
(p,q),\, (r,p), \, (q,r) \; \to \;  (q,r),\, (p,q), \, (r,p) \;  \to \; (r,p),\, (q,r), \, (p,q) 
\]
correspond to the following permutations
\[
(p,q,r) \; \to \; (q,r,p) \; \to \; (r,p,q).
\]

Let $X^\dag$ denote the conjugate transpose of $X$. Since $\lvert \tr(XMY) \rvert = \lvert \tr( Y^\dag M^\dag  X^\dag)\rvert$ and $\|X\|_{p,q} = \|X^\dag\|_{q^*,p^*}$, we have
\[
\frac{|\tr(XMY)|}{\|X\|_{p,q}\|Y\|_{q,r}\|M\|_{r,p}} = 
\frac{| \tr( Y^\dag M^\dag  X^\dag)|}{\|Y^\dag\|_{r^*,q^*}\|X^\dag \|_{q^*,p^*}\|M^\dag\|_{p^*,r^*}}.
\]
Taking maximum over all nonzero $X, Y, M$ yields the required equality. Note that the proof works over both $\mathbb{R}$ and $\mathbb{C}$.
\end{proof}
A straightforward application of H\"older's inequality yields an upper bound for $\lVert\mu_{l,m,n}\rVert_{p,q,r}$.
\begin{theorem}\label{thm:upper}
Let $p,q,r\in[1,\infty]$ and $l,m,n\in\mathbb{N}$. For any nonzero matrices $X\in\mathbb{F}^{l\times m},Y\in\mathbb{F}^{n\times l}$ and $M\in\mathbb{F}^{m\times n}$, the following inequality is sharp:
\begin{equation}\label{Gro_inq_lemma}
\frac{|\tr(XMY)|}{\|X\|_{p,q}\|Y\|_{q,r}\|M\|_{r,p}}\le  \frac{|\tr(XMY)|}{\|X\|_{1,2}\|Y\|_{2,\infty}\|M\|_{\infty,1}}\cdot l^{|1/q-1/2|}\cdot m^{1-1/p}\cdot n^{1/r}.
\end{equation}
Furthermore, we have a generalization of Grothendieck's inequality:
\begin{equation}\label{Gro_inq_pqr}
\lVert\mu_{l,m,n}\rVert_{p,q,r}=\max_{X,Y,M\neq0}\frac{|\tr(XMY)|}{\|X\|_{p,q}\|Y\|_{q,r}\|M\|_{r,p}}\le  K_G^\mathbb{F}\cdot l^{|1/q-1/2|}\cdot m^{1-1/p}\cdot n^{1/r}.
\end{equation}
\end{theorem}
\begin{proof}
First let $1\le  q\le 2$. H\"{o}lder's inequality together with the  fact that $\|x\|_p\le \|x\|_q$ whenever $q\le  p$ give us
\begin{equation}\label{proof:inq0}
\lVert X\rVert_{1,2}\le \lVert X\rVert_{1,q}\le \lVert X\rVert_{p,q}\qquad\text{and}\qquad
\lVert Y\rVert_{2,\infty}\le \lVert Y\rVert_{2,r}\le  l^{1/q-1/2}\lVert Y\rVert_{q,r}.
\end{equation}
The same argument also gives $\|M\|_{\infty,p}\le \|M\|_{\infty,1}\le  m^{1-1/p}\|M\|_{\infty,p}$ for $1\le  p\le \infty$
and thus
\begin{equation}\label{proof:inq1}
 \lVert M\rVert_{\infty,1}\le  m^{1-1/p}\lVert M\rVert_{\infty,p}\le  n^{1/r}\cdot m^{1-1/p}\lVert M\rVert_{r,p}.
\end{equation}
\eqref{Gro_inq_lemma} then follows from \eqref{proof:inq0} and \eqref{proof:inq1}.
To see that it is sharp, we use the following\footnote{These are standard in matrix theory, often used to demonstrate sharpness of various matrix inequalities.} $m\times n$ rank-one matrices:%
{\footnotesize
\[
E_{m,n}\coloneqq  \begin{bmatrix}
1 & 0 &\dots & 0\\
0 & 0 &\dots & 0\\
\vdots &\vdots & & \vdots\\
0 & 0 &\dots & 0
\end{bmatrix},
\;\;
C_{m,n}\coloneqq\begin{bmatrix}
1 & 0 &\dots & 0\\
1 & 0 &\dots & 0\\
\vdots &\vdots & & \vdots\\
1 & 0 &\dots & 0
\end{bmatrix},
\;\;
R_{m,n}\coloneqq\begin{bmatrix}
1 & 1 &\dots & 1\\
0 & 0 &\dots & 0\\
\vdots &\vdots & & \vdots\\
0 & 0 &\dots & 0
\end{bmatrix},
\;\;
J_{m,n}\coloneqq\begin{bmatrix}
1 & 1 &\dots & 1\\
1 & 1 &\dots & 1\\
\vdots &\vdots & & \vdots\\
1 & 1 &\dots & 1
\end{bmatrix}.
\]}%
It is easy to check that
\begin{alignat*}{4}
\|E_{l,m}\|_{p,q}&=1, \qquad &\|R_{n,l}\|_{q,r}&=l^{1-1/q},  \qquad &\|J_{m,n}\|_{r,p}&=m^{1/p}\cdot n^{1-1/r},\\
\|E_{l,m}\|_{1,2}&=1,  \qquad &\|R_{n,l}\|_{2,\infty}&=l^{ 1/2}, \qquad &\|J_{m,n}\|_{\infty,1}&=mn.
\end{alignat*}
Since \eqref{Gro_inq_lemma} becomes an equality when $X=E_{l,m}$, $Y=R_{n,l}$, and $M=J_{m,n}$, it is sharp for   $1\le  q\le 2$.

Next let $2<q\le \infty$. Similarly, we have
\[
l^{1/q-1/2}\lVert X\rVert_{1,2}\le \lVert X\rVert_{1,q}\le \lVert X\rVert_{p,q}\qquad\text{and}\qquad
\lVert Y\rVert_{2,\infty}\le \lVert Y\rVert_{2,r}\le  \lVert Y\rVert_{q,r},
\]
which together with \eqref{proof:inq1} give us \eqref{Gro_inq_lemma}. In this case the sharpness follows from
\begin{alignat*}{4}
\|C_{l,m}\|_{p,q}&=l^{1/q},\qquad &\|E_{n,l}\|_{q,r}&=1,\qquad & \|J_{m,n}\|_{r,p}&=m^{1/p}\cdot n^{1-1/r},\\
\|C_{l,m}\|_{1,2}&=l^{1/2}, \qquad & \|E_{n,l}\|_{2,\infty}&=1, \qquad & \|J_{m,n}\|_{\infty,1}&=mn,
\end{alignat*}
and selecting $X=C_{l,m}$, $Y=E_{n,l}$, and $M=J_{m,n}$.

\eqref{Gro_inq_pqr} follows from taking maximum over nonzero $X,M,Y$ and supremum over $l,m,n$. When $(p,q,r)=(1,2,\infty)$, it yields  Grothendieck's inequality \eqref{Gro_norm_combine}.
\end{proof}

The upper bound in \eqref{Gro_inq_pqr} depends on $l,m,n$ except when $(p,q,r)$ is  $(1,2,\infty)$ or a cyclic permutation (by Lemma~\ref{lem:simple}\eqref{it:cyclic}). An immediate question is whether a uniform bound independent of $l,m,n$  might perhaps also exist for some other values of $(p,q,r)$, i.e.,
\begin{equation}\label{Gro_inq_general}
K_{p,q,r} \coloneqq  \sup_{l,m,n\in\mathbb{N}}\lVert\mu_{l,m,n}\rVert_{p,q,r}  < \infty?
\end{equation}
In Section~\ref{sec:unique}, we will see that $K_{p,q,r}  = \infty$ for all $(p,q,r) \notin  \{ (1,2,\infty), \; (\infty, 1, 2), \; (2, \infty, 1)\}$. Nevertheless, we stress that the absence of a uniform bound is only limited to the class of $(p,q,r)$-norms in \eqref{eq:norm}. For example, we may consider the tensor \emph{spectral norm} \cite{nuclear} of  $\mu_{l,mn,n}$,
\[
\lVert \mu_{l,m,n}\rVert_{\sigma}\coloneqq \max_{X,Y,M\neq0}\frac{|\tr(XMY)|}{\lVert X\rVert_{F}\lVert Y\rVert_{F}\lVert,
M\rVert_{F}}
\]
where the norm on $X,Y,M$ is the  matrix Frobenius (i.e., Hilbert--Schmidt) norm. In this case,
\begin{equation}\label{eq:spectral}
\lVert \mu_{l,m,n}\rVert_{\sigma}=1, \qquad \text{for all} \; l,m,n \in \mathbb{N},
\end{equation}
since, by Cauchy--Schwartz and  the submultiplicativity of the Frobenius norm,
\[
|\operatorname{tr}(XMY)|\le \lVert X\rVert_{F}\lVert MY\rVert
_{F}\le \lVert M\rVert_{F}\lVert X\rVert_{F}\lVert Y\rVert_{F},
\]
and equality is attained by choosing $M,X,Y$ with $1$ in the $(1,1)$th entry and $0$ everywhere else.

We will use \eqref{eq:spectral} to obtain lower bounds on $\lVert \mu_{l,m,n}\rVert_{p,q,r}$ below.  \eqref{Gro_inq_pqr} and \eqref{eq:lower} will collectively be referred to as the \emph{Grothendieck--H\"older inequality}. 
\begin{theorem}\label{thm:lower}
Let $p,q,r\in[1,\infty]$ and $l,m,n\in\mathbb{N}$. Then
\begin{equation}\label{eq:lower}
\frac{1}{l^{|1/q-1/2|}\cdot m^{|1/p-1/2|}\cdot n^{|1/r-1/2|}}\le \lVert \mu_{l,m,n}\rVert_{p,q,r}.
\end{equation}
\end{theorem}
\begin{proof}
For $n\in\mathbb{N}$ and $p,q\in[1,\infty]$, let
\[
c_{p,q}(n)\coloneqq n^{\max\{0, 1/p-1/q\}}.
\]
Then for any $M\in\mathbb{F}^{m\times n}$, the following sharp inequality holds \cite[Theorem~4.3]{Klaus}, 
\[
\lVert M\rVert_{p,q}\le  c_{q,2}(m)c_{2,p}(n)\lVert M\rVert_{F}.
\]
It follows that
\[
\lVert X\rVert_{p,q}\le  c_{q,2}(l)c_{2,p}(m)\lVert X\rVert_{F}, \quad \lVert Y\rVert_{q,r}\le  c_{r,2}(n)c_{2,q}(l)\lVert Y\rVert_{F},\quad
\lVert M\rVert_{r,p}\le  c_{p,2}(m)c_{2,r}(n)\lVert M\rVert_{F},
\]
and for any tensor $\tau \in  (\mathbb{F}^{l \times m})^* \otimes  (\mathbb{F}^{m \times n})^* \otimes  (\mathbb{F}^{n \times l})^*$, we have
\[
\lVert \tau \rVert_{\sigma}\le \lVert \tau\rVert_{p, q, r}\cdot l^{|1/q-1/2|}\cdot m^{|1/p-1/2|}\cdot n^{|1/r-1/2|}.
\]
Plugging in $\tau = \mu_{l,m,n}$ and using \eqref{eq:spectral},  we obtain \eqref{eq:lower}.
\end{proof}

A practical reason for wanting to ascertain \eqref{Gro_inq_general} is that if
\begin{equation}\label{eq:poly}
(p,q)\; \text{and} \; (q,r) \in \{ (1,1),\; (2,2),\; (\infty,\infty),\; (1,q),\; (q,\infty) \},
\end{equation}
then $\lVert X\rVert_{p,q}$ and $\lVert Y\rVert_{q,r}$ can be computed in polynomial time (to arbitrary precision)  and
\[
\max_{X,Y\neq0}\frac{|\operatorname{tr}(XMY)|}{\lVert X\rVert_{p,q}\lVert
Y\rVert_{q,r}}\le  K_{p,q,r}\lVert M\rVert_{r,p}%
\]
in principle gives a polynomial-time approximation of $\lVert M\rVert_{r,p}$, which is NP-hard \cite{HO10} if $(r,p)$ is not one of the special cases in \eqref{eq:poly}. Unfortunately, we now know that as $K_{p,q,r} = \infty$ in all other cases, this only works when $(p,q,r) \in  \{ (1,2,\infty), \; (\infty, 1, 2), \; (2, \infty, 1)\}$, all three of which are equivalent to Grothendieck's inequality.

\section{Grothendieck's inequality is unique}\label{sec:unique}

We show that $(p,q,r)=(1,2,\infty)$ is, up to a cyclic permutation, the only case for which \eqref{Gro_inq_general} holds. We will first rule out a large number of cases with the following proposition.
\begin{proposition}\label{pqr_identitymatrix}
Let $p,q,r\in[1,\infty]$. If there exists a finite constant $K_{p,q,r} > 0$ such that
$\lVert\mu_{l,m,n}\rVert_{p,q,r} \le K_{p,q,r}$ for all $l, m, n \in \mathbb{N}$, then
\[
\min (p,q,r) =1 \qquad \text{and} \qquad \max (p,q,r) =\infty.
\]
\end{proposition}
\begin{proof}
Let $I_{m,n}\in\mathbb{F}^{m\times n}$ be the matrix obtained by appending zero rows or columns to the identity matrix\footnote{Note that $I_{n,n} = I_n$. For consistency, we will always use the latter notation when it is a square matrix.} $I_n$ or $I_m$,
\[
I_{m,n} \coloneqq 
\begin{cases}
[I_{n}, 0_{m - n}]^\tp
& \text{if}\; m\ge  n, \\
[I_{m},0_{n - m}]
& \text{if}\; m<n.
\end{cases}
\]
Then its matrix $(p,q)$-norm  is
\begin{equation}\label{norm:identity}
\lVert I_{m,n}\rVert_{p,q}=
\begin{cases}
\min\{m,n\}^{1/q-1/p} & \text{if}\; p\ge  q,\\
1 & \text{if}\; p<q.
\end{cases}
\end{equation}
This follows from an easy calculation using H\"older inequality: For $m\ge  n$,
\[
\lVert I_{m,n} \rVert_{p,q}=\max_{z\neq0}\frac{\lVert I_{m,n}z \rVert_q}{\lVert z\rVert_p}=\max_{z\neq0}\frac{\lVert z \rVert_q}{\lVert z\rVert_p}=\begin{cases}
n^{1/q-1/p} &\text{if}\; p\ge  q,\\
1 &\text{if}\; p<q,
\end{cases}
\]
and for $m<n$,
\[
\lVert I_{m,n}\rVert_{p,q}=\max_{z\neq0}\frac{\lVert I_{m,n}z \rVert_q}{\lVert z\rVert_p}=\max_{z\neq0}\frac{\lVert z_m \rVert_q}{\lVert z\rVert_p}=\max_{z_m\neq0}\frac{\lVert z_m \rVert_q}{\lVert z_m\rVert_p}=\begin{cases}
m^{1/q-1/p} &\text{if}\; p\ge  q,\\
1  &\text{if}\; p<q,
\end{cases}
\]
where $z_m=[z_1,\dots,z_m]\in\mathbb{F}^m$ is the vector comprising the first $m$ entries of $z$.

Set $X=I_{l, m}$, $Y=I_{n,l}$, and $M=I_{m, n}$. Then  $\tr(XMY)=\min\{l,m,n\}$, and by \eqref{norm:identity}, we obtain
\[
\frac{|\tr(XMY)|}{\|X\|_{p,q}\|Y\|_{q,r}\|M\|_{r,p}}=\begin{cases}
\min\{l,m,n\}\cdot\min\{m,n\}^{1/r-1/p} &\text{if}\; p\le  q\le  r,\\
\min\{l,m,n\}\cdot \min\{l,n\}^{1/q-1/r}\cdot\min\{m,n\}^{1/r-1/p} & \text{if}\; p\le  r\le  q,\\
\min\{l,m,n\}\cdot \min\{l,m\}^{1/p-1/q}\cdot\min\{m,n\}^{1/r-1/p} &\text{if}\; q\le  p\le  r,\\
\min\{l,m,n\}\cdot \min\{l,m\}^{1/p-1/q} &\text{if}\; q\le  r\le  p,\\
\min\{l,m,n\}\cdot\min\{l,n\}^{1/q-1/r} &\text{if}\; r\le  p\le  q,\\
\min\{l,m,n\}\cdot \min\{l,m\}^{1/p-1/q}\cdot \min\{l,n\}^{1/q-1/r}  &\text{if}\; r\le  q\le  p.
\end{cases}
\]
Suppose $l=2n$, $m=n$ and $p\le  q\le  r$, then
\[
\lim_{n\rightarrow\infty}\frac{|\tr(XMY)|}{\|X\|_{p,q}\|Y\|_{q,r}\|M\|_{r,p}}=\lim_{n\rightarrow\infty}
n^{1/r-1/p+1}=\infty
\]
unless $p=1$ and $r=\infty$. Repeating the argument for all possible permutations of $(p,q,r)$ and taking advantage of  Lemma~\ref{lem:simple}\eqref{it:cyclic}, we conclude that $\min (p,q,r) =1$ and $\max (p,q,r) =\infty$ is necessary for the uniform boundedness of $\lVert\mu_{l,m,m}\rVert_{p,q,r}$.
\end{proof}

We will next eliminate the remaining possibilities. Our approach will rely on the existence of Hadamard matrices of arbitrarily large dimensions. Indeed, an $n \times n$ Hadamard matrix $H_n \in\{\pm 1\}^{n\times n}$ exists for any $n$  divisible by $4$, or, for concreteness, we may set $H_n = \begin{bsmallmatrix} 
1&1\\1&-1\end{bsmallmatrix}^{\otimes k}$ with $n = 2^k$ \cite[Section~2.1]{Hadamard}. The bottom line is that we may let $n \to \infty$ in the proof below.
\begin{theorem}[Uniqueness of Grothendieck's inequality]\label{thm:unique}
Let $1 \le p,q,r\le \infty$ and $l,m,n\in\mathbb{N}$. Then $\lVert\mu_{l,m,n}\rVert_{p,q,r}$ is uniformly bounded for all $l, m,n \in \mathbb{N}$ if and only if
\[
(p,q,r) \in  \{ (1,2,\infty), \; (\infty, 1, 2), \; (2, \infty, 1)\}.
\]
\end{theorem}
\begin{proof}
We will see that it suffices to take $l=m=n$ throughout this proof. 
By Lemma~\ref{lem:simple}\eqref{it:cyclic} and Proposition~\ref{pqr_identitymatrix}, we may assume that $p=1$ and either $q=\infty$ or $r=\infty$.  We will show that $\tr(XMY)$ is unbounded for judiciously chosen $n \times n$ real matrices $X$, $M$, and $Y$ as $n \to \infty$. 

\myuline{\textsc{Case I:} $(1, q, \infty)$, $1 \le q\le \infty$.}  Suppose $2 < q \le \infty$.   Let  $X_0 =n^{-1/q} \Delta$ for some arbitrary $\Delta = (\delta_{ij}) \in\{\pm 1\}^{n\times n}$ and let $Y_0 =I_n$.  Then $\|X_0\|_{1,q}= n^{-1/q} \|\Delta\|_{1,q} = 1$ and  $\|Y_0\|_{q,\infty}=\|I_n\|_{q,\infty}=1$ by \eqref{eq:pqnorm}. For any $M = (M_{ij})\in\mathbb{R}^{n\times n}$,
\[
\max_{\|X\|_{1,q}, \;\|Y\|_{q,\infty}\le 1} \lvert \tr(XMY) \rvert \ge   \lvert \tr(X_0MY_0) \rvert = n^{-1/q} \lvert \tr(\Delta M) \rvert
= n^{-1/q} \Bigl\lvert \sum\nolimits_{i,j=1}^n \delta_{ij} M_{ij} \Bigr\rvert.
\]
Since $\Delta\in\{\pm 1\}^{n\times n}$ is arbitrary, we will choose $\delta_{ij}$ so that $ \delta_{ij} M_{ij}$ is nonnegative. Thus
\begin{equation}\label{eq:XY}
\max_{\|X\|_{1,q}, \;\|Y\|_{q,\infty}\le 1} \lvert \tr(XMY) \rvert  \ge  n^{-1/q}\sum_{i,j=1}^n |M_{ij}|.
\end{equation}
Let $H_n \in\{\pm 1\}^{n\times n}$ be a Hadamard matrix. So $H_nH_n^\tp  =n I_n$ and all singular values of $H_n$ are $\sqrt{n}$ \cite{FriAli}. Therefore, by \eqref{eq:inftyone},
\begin{equation}\label{eq:Hnorm}
\|H_n\|_{\infty,1} = \max_{\varepsilon, \delta \in \{\pm 1\}^n} \lvert \varepsilon^\tp H_n \delta \rvert \le \sigma_{\max} (H_n) \lVert \varepsilon\rVert_2 \lVert \delta\rVert_2  = n^{3/2}.
\end{equation}
Let $M = n^{-3/2} H_n$. Then  $\|M \|_{\infty,1}\le 1$ and by \eqref{eq:XY},
\[
\max_{\|X\|_{1,q}, \;\|Y\|_{q,\infty}, \; \|M \|_{\infty,1} \le 1} \lvert \tr(XMY)  \rvert \ge
 n^{-1/q} \times n^{-3/2} \times n^2 = n^{1/2- 1/q} \to \infty
\]
as $n \to \infty$.

Suppose $1 \le q < 2$. Since the H\"older conjugates are  $r^*=1$, $2 < q^* \le \infty$, and $p^*=\infty$, by Lemma~\ref{lem:simple}\eqref{it:conjugate}, this reduces to the case we just treated.
 
\myuline{\textsc{Case II:} $(1, \infty, r)$, $1 \le r \le \infty$.}  For  $r = \infty$, we have $(1,\infty, \infty)$, which is same as the $q= \infty$ case in \textsc{Case I}. For  $r = 1$, we have $(1,\infty, 1)$, but by Lemma~\ref{lem:simple}\eqref{it:cyclic}, this is equivalent to $(1,1,\infty)$, which is same as the $q= 1$ case in \textsc{Case I}. So we may assume $1 < r < \infty$.
 
Suppose $1 < r  < 2$.  Let $M=n^{1/r-1}I_n$ and $Y=n^{-3/2} H_n$ where $H_n \in\{\pm 1\}^{n\times n}$  is a Hadamard matrix. Then $\|M\|_{r,1} = n^{1/r-1} \|I_n\|_{r,1}= 1$ by \eqref{norm:identity}, and  $\|Y\|_{\infty,r}\le \|Y\|_{\infty,1} =n^{-3/2}\| H_n \|_{\infty,1} \le 1$ by \eqref{eq:Hnorm}.  We choose $X\in\{\pm 1\}^{n\times n}$ such that $\tr (XH_n)=n^2$ and thus $\tr(XY) = n^{-1/2}$. Clearly $\|X\|_{1,\infty} = 1$ by \eqref{eq:pqnorm}. Hence
\[
\tr (XMY) = n^{1/r-1}\tr(XY) =   n^{1/r-1/2} \to \infty
\]
as $n\to\infty$.

Suppose $2< r  < \infty$.  Since the H\"older conjugates are  $1 < r^*  < 2$, $q^* =1$, and $p^*=\infty$, by Lemma~\ref{lem:simple}\eqref{it:conjugate} , this is equivalent to the case $(r^*, 1, \infty)$. Now by  Lemma~\ref{lem:simple}\eqref{it:cyclic}, this is in turn equivalent to the case $(1, \infty, r^*)$ with  $1 < r^*  < 2$, which is the case we just treated.

Suppose $r=2$.  Let $Y=n^{-1} H_n$ where $H_n \in\{\pm 1\}^{n\times n}$ is again a Hadamard matrix.  Then
\[
\|H_n \|_{\infty,2}=\max_{ x\in\{\pm 1\}^{n}} \|H_n x\|_2 \le \sigma_{\max}(H_n)\sqrt{n}=n.
\]
So $\|Y\|_{\infty,2}\le 1$. Let $M=n^{-1/2} I_n$.  Then $\|M\|_{2,1}=1$ by \eqref{norm:identity}.  Let $X\in\{\pm 1\}^{n\times n}$ be such that $\tr( XH_n) = n^2$ and thus $\tr(XY) = n$. Clearly $\|X\|_{1,\infty} = 1$ by \eqref{eq:pqnorm}. We have
\[
\tr( XMY) = n^{-1/2}  \tr( XY) = n^{1/2} \to \infty
\]
as $n \to \infty$.
\end{proof}

\section{Conclusion}

We hope our characterization of Grothendieck's constant as a norm of the central object in the study of fast matrix multiplications would spur interactions between the two areas and perhaps even facilitate the determination of its exact value. Knowing that Grothendieck's inequality is a unique instance within a family of natural norm inequalities may help us better understand its ubiquity and utility. In fact, the way we formulate Grothendieck's inequality in \eqref{eq:gro1} facilitated our elementary proof of the inequality in \cite{ZFL}, which is one that works over both (i) $\mathbb{R}$ and (ii) $\mathbb{C}$ and yields both (iii) Krivine's bound and (iv) Haagerup's bound.

\subsection*{Acknowledgment}

The work in this article is generously supported by DARPA D15AP00109 and NSF IIS 1546413. LHL gratefully acknowledges the support of a DARPA Director's Fellowship and the Eckhardt Faculty Fund.

\bibliographystyle{abbrv}

\end{document}